\documentclass[english,a4paper]{article}
\usepackage{dcolumn}
\usepackage{bm}
\usepackage{graphicx}
\usepackage{amsmath}
\usepackage{amsfonts}
\usepackage{amssymb}
\usepackage{amsthm}
\usepackage{amstext}
\usepackage{amsbsy}
\usepackage{amsopn}
\usepackage{amscd}
\usepackage{amsxtra}

\newtheorem{theorem}{Theorem}

\newtheorem{definition}{Definition}

\def\phi{\varphi}
\def\epsilon{\varepsilon}

\begin{document}
\title{Optimizing fingerprinting experiments for parameter identification: Application to spin systems}
\author{Q. Ansel, M. Tesch, S. J. Glaser\footnote{Department of Chemistry, Technische Universit\"at
M\"unchen, Lichtenbergstrasse 4, D-85747 Garching, Germany}, D. Sugny\footnote{Laboratoire Interdisciplinaire Carnot de
Bourgogne (ICB), UMR 6303 CNRS-Universit\'e Bourgogne-Franche Comt\'e, 9 Av. A.
Savary, BP 47 870, F-21078 Dijon Cedex, France and Institute for Advanced Study, Technische Universit\"at M\"unchen, Lichtenbergstrasse 2 a, D-85748 Garching, Germany, dominique.sugny@u-bourgogne.fr}}

\maketitle

\begin{abstract}
We introduce the Optimal Fingerprinting Process which is aimed at accurately identifying the parameters which characterize the dynamics of a physical system. A database is first built from the time evolution of an ensemble of dynamical systems driven by a specific field, which is designed by optimal control theory to maximize the efficiency of the recognition process. Curve fitting is then applied to enhance the precision of the identification. As an illustrative example, we consider the estimation of the relaxation parameters of a spin- 1/2 particle. The experimental results are in good accordance with the theoretical computations. We show on this example a physical limit of the estimation process.
\end{abstract}


\section{Introduction}
The fingerprinting method is a well-known technique generally used for determining the identity of a person. The basic concept is illustrated in Fig.~\ref{fig:fingerprint_principle}. The overall process can be decomposed into three different steps: \textit{(i)} a fingerprint recording of an ensemble of subjects, \textit{(ii)} the creation of a database (also called dictionary) where fingerprint images are associated with person identities, and \textit{(iii)} a recognition process where a numerical search algorithm finds the closest database element to the fingerprint of an unknown subject. Assuming that fingerprints are different for each person, a mapping between fingerprints and persons can be defined, making possible the identification protocol.
\begin{figure}
	\includegraphics[width=1.0\linewidth]{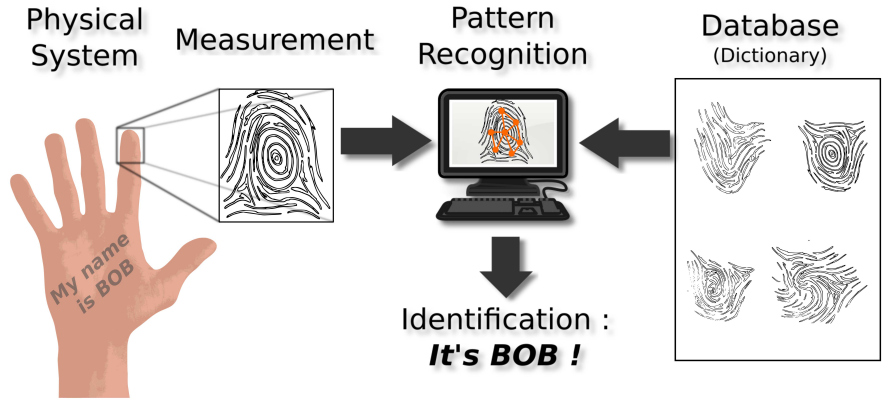}
	\caption{(Color online) The fingerprinting principle: A physical system is identified by using a pattern recognition between a fingerprint measurement and the elements of a database.}\label{fig:fingerprint_principle}
\end{figure}
This idea can be generalized to any system which has unique properties that can be revealed by a measurement process. This approach can be applied in a static setting, but also in a dynamical one where the system is subjected to an external control field. In this latter case, each element of the database corresponds to the time evolution of some observables under the action of the field, thus increasing the complexity of the fingerprints and the precision of the estimation. This idea has been recently adapted to Magnetic Resonance Imaging (MRI) for the identification of tissue parameters \cite{naturemri}. This initial investigation led to an impressive number of studies in this domain (see e.g. \cite{naturemri2,jacobpaper} to cite a few). A crucial issue in this fingerprinting process is the design of the excitation field. A simple approach using a time-dependent random field was proposed in \cite{naturemri} in order to limit the correlations between the different fingerprints. However, this approach does not incorporate information about the system dynamics and the recognition process and is therefore not expected to reach its precision limit. To overcome this fundamental difficulty, we propose in this paper to combine a standard fingerprinting process with recently developed optimal control techniques in quantum control \cite{rabitzreview,quaintreview,grape}. This Optimal Fingerprinting Process (OFP) allows us to maximize the efficiency of the identification and to minimize the errors made in the parameter estimation. As an illustrative example, this method will be used to identify the relaxation parameters of a spin- 1/2 particle. In this case, the estimation is made from a series of free-induction decay signals induced by impulsive excitations of different intensities. Note also that the measurements we consider are classical as they result from a continuous measurement of a large number of quantum systems. Nevertheless, the same type of processes \cite{haidong_metrology,haidong_metrology2} could be carried out with quantum measurements in quantum metrology \cite{quantummetrology1,quantummetrology2}. Closely related but different concepts using external fields to estimate the parameters of quantum systems have also been developed in the past few years \cite{geremia1,geremia2,turinici,schirmer1,schirmer2,schirmer3,li1,li2}. The combination of optimization and fingerprinting techniques has been explored in recent works in MRI~\cite{jacobpaper,maidens}.

The paper is organized as follows. Section~\ref{sec2} describes the theoretical framework of the method. The technique is applied in Sec.~\ref{sec3} to estimate the relaxation parameters of a spin system. The stability of the estimation in presence of noise is discussed and analyzed numerically. The efficiency of the optimal fingerprinting process is demonstrated experimentally in Sec.~\ref{sec4} on a spin 1/2 particle by using techniques of Nuclear Magnetic Resonance. Conclusion and prospective views are given in Sec.~\ref{conc}. A mathematical description of the method and numerical results are reported respectively in the Appendices \ref{appa}, \ref{appb} and \ref{appc}.
\section{Theoretical framework}\label{sec2}
We start the analysis with a general presentation of the method on an abstract system. We refer the reader to the Appendix~\ref{appa} for a detailed mathematical description of this technique. The state of the system is given at a time $t$ by $\Psi(t)\in\mathcal{H}$ ($\mathcal{H}$ is generally a Hilbert space), whose dynamics are governed by the following differential equation:
\begin{equation}\label{eq1}
\dot{\Psi}=\hat H (\vec{S},u(t))\Psi,
\end{equation}
where $\hat H$ is a linear operator (the Hamiltonian for quantum systems), $\vec{S}\in\mathbb{R}^p$ defines the $p$ unknown parameters to estimate and $u(t)$ is the control field. 
For a generic control field, the resulting time evolution $\Psi(t)$ will be different for each system characterized by different $\vec{S}$. The experimental system returns a specific response $g(\vec{S}_{0},t)$. This response defines the fingerprint of the system and is assumed to be unique for a given vector $\vec{S}_{0}$. In this measurement process, note that experimental noise has to be accounted for. This point will be investigated in the section about the case study in NMR.

The database is built from numerical simulations of the time evolution of $N$ systems characterized by specific values of the $\vec{S}$ parameters, denoted $\{\vec{S}_n\}_{n=1,\cdots, N}$. The database is defined as a set $\{f_n(t)\}_{n=1,\cdots ,N}$ of $N$ real functions associated with each $\vec{S}_n$.
%
%
The vector $\vec{S}_{0}$ is determined from the best match between $g(t)$ and one of the elements $f_n(t)$ of the dictionary. This leads to the estimation $\vec{S}_{0}\simeq \vec{S}_k$ for the experimental system, where $k\in \{1,\cdots , N\}$. The match is performed by minimizing the functional $D$ (called the \textit{recognition map}) over the elements of the dictionary. $D$ is defined as follows:
\begin{equation}\label{cost1}
D[f_n,g]=\left\Vert \frac{f_n}{||f_n||} - \frac{g}{||g||}\right\Vert^2.
\end{equation}
In Eq.~\eqref{cost1}, the two vectors are divided by their norm to eliminate a possible scaling factor between the experimental data and the theoretical model.

At this point, the method can be performed for any control field that distinguishes the elements of the dictionary. OFP is defined by introducing a 
figure of merit which is aimed at maximizing the distance between the elements of the dictionary, and thus improving the recognition process and the precision of the method. This functional $C$ can be expressed as:
\begin{equation}
C_N = \frac{1}{2N^2}\sum_{m,n} \mu_{m n} D\left[f_m,f_n \right],
\end{equation}
where the $\mu_{m n}$ are some weight factors. The parameter $C_N$ is the normalized average distance between the $N$ elements of the dictionary. The normalization factor ensures that $C_N\leq 1$ when all weights $\mu_{m n}=1$ (see Appendix~\ref{appa}). Maximizing this quantity allows the minimization of the overlap of $g$ with the other functions of the dictionary and thus the error made in the estimation procedure. As described in Appendix~\ref{appb}, a generalized version of the optimal control GRAPE algorithm \cite{grape,trackinggrape} can be used to numerically generate the control field. We stress that this extension is not trivial since the maximization of the distance between the different systems is not made at one point but for the whole time evolution. OFP is schematically described in Fig.~\ref{fig:map}.
\begin{figure}[h]
	\includegraphics[width=1.0\linewidth]{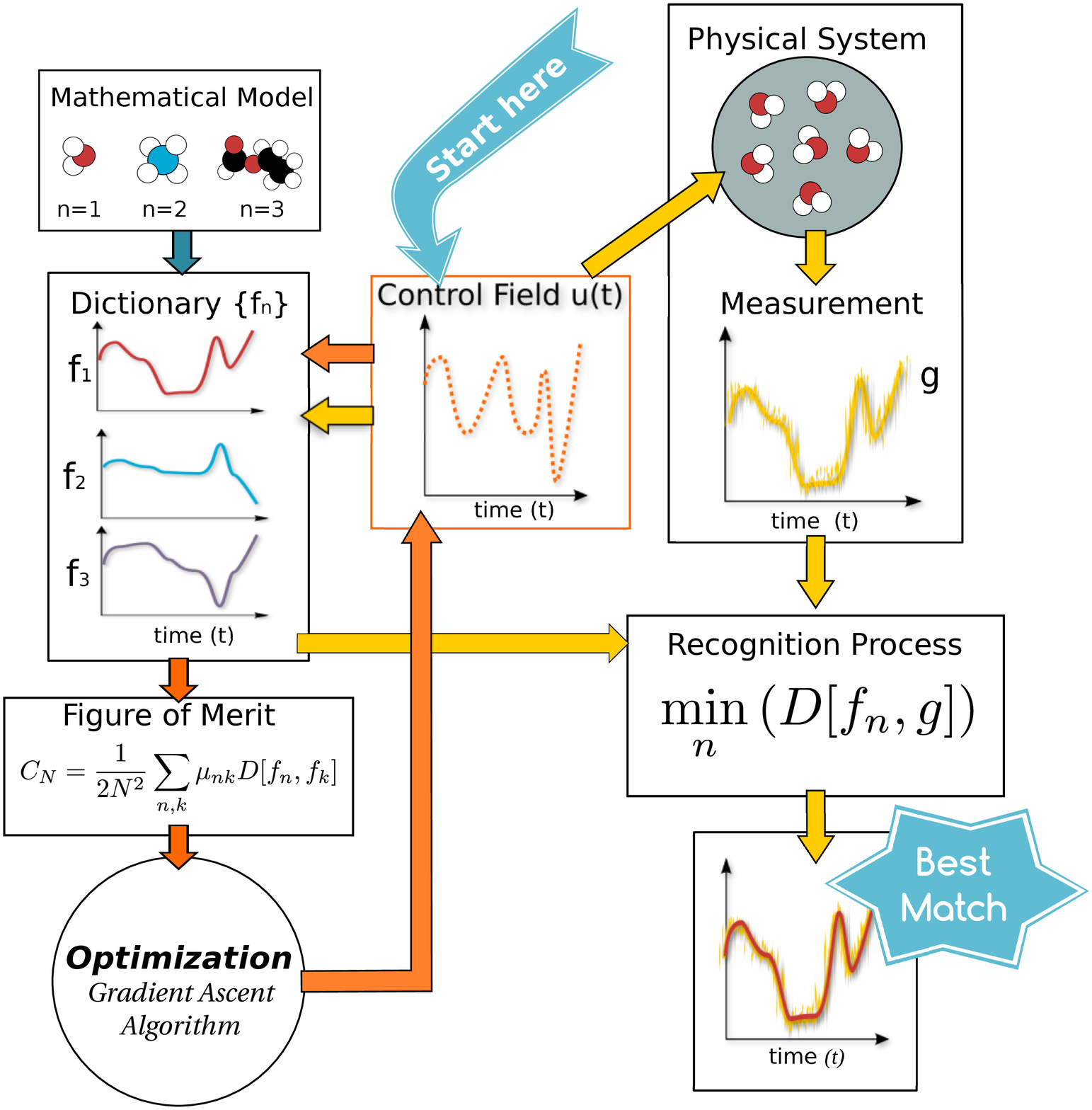}
	\caption{(Color online) The OFP is composed of two different loops. The first loop (yellow or light gray arrows) is the standard fingerprinting process. A control field $u(t)$ is designed at the starting point of the loop. This field is applied to a physical system which returns a specific response $g(t)$. On the other side, the response is computed numerically for an ensemble of physical systems with different values of the parameters. These simulations define a \textit{dictionary} of functions $f_n(t)$. The recognition process allows us to find the best match between elements of the dictionary $f_n$ and the result of the measurement $g$ (see the text for details). The second loop (orange or dark gray arrows) describes the dictionary optimization. The optimization is performed for an ensemble of $N$ systems with different values of the parameters. An optimal control algorithm is used to maximize numerically the figure of merit $C_N$.\label{fig:map}}
\end{figure}
%

\section{A case study in Nuclear Magnetic Resonance}\label{sec3}
As an illustrative example, we investigate the estimation of the relaxation parameters of a spin system by OFP in NMR~\cite{ernstbook,levittbook}. We consider an inhomogeneous ensemble of spin- 1/2 particles with different resonance offsets $\omega$ and radio-frequency inhomogeneities $\alpha$ whose dynamics are ruled by the Bloch equations. In the rotating frame, the equation of motion for each isochromat is given by~\cite{levittbook,kozbar}:
\begin{equation}\label{eq:bloch}
\frac{d}{dt} \vec M
    =
    \left(
    \begin{array}{cccc}
    1 & 0&0&0\\\
    0             & -\frac{1}{T_2} & -\omega     & \alpha\omega_y(t) \\
    0             & \omega     & -\frac{1}{T_2} &  -\alpha\omega_x(t) \\
    \frac{1}{T_1} & -\alpha\omega_y(t)  & \alpha\omega_x(t) & -\frac{1}{T_1}
    \end{array} \right)
    \vec M
\end{equation}
where $\vec M = (1,M_x,M_y,M_z)^t$ is the extended Bloch vector (the radius of the Bloch ball is normalized to 1) and $M_{x,y,z}^{(\omega)}$ its coordinates along the $x,y,z$- directions. In this case, note that $\mathcal{H} = \mathbb{R}^4$. The relaxation times $T_1$ and $T_2$ are assumed to be the same for all the isochromats of the sample. The control amplitudes are given by $\omega_x(t)$ and $\omega_y(t)$. $\omega$ is the resonance offset and the parameter $\alpha$ describes the experimental scaling of the radio-frequency field applied to the sample~\cite{kozbar}. In this example, OFP is used to estimate the parameters $T_1$ and $T_2$, and the measured signal results from the average magnetization of the spins with different values of $\omega$ and $\alpha$. The functions of the dictionary are given by
\begin{equation}\label{eqdicspin}
f_n(t) = \big(\bar M_x^{(n)}(t),\bar M_y^{(n)}(t) \big),
\end{equation}
where $\bar M_\mu$, $\mu=x,y$, is the average of $M_\mu$ over the sample. 
The averaging procedure is defined by using a probability distribution in $\omega$ and $\alpha$ which can either be known before the optimization of the database or adjusted during the recognition process (see below for an example).

The control field is a sequence of short pulses (with respect to $T_1$ and $T_2$), modeled by Dirac distributions:
\begin{align}
\omega_\mu(t) = \sum_{k=1}^{N_p} \omega_{\mu,k} \delta (t-kT), && \mu = x,y,
\end{align}
where $\omega_{\mu,k}$ is the amplitude of the $k$-th pulse, $N_p$ the number of pulses and $T$ the time between each pulse. 
This approximation leads to a straightforward time discretization of the dynamics of the system. The measured signal corresponds to the average magnetization just after each $\delta$- pulse, $\big(\bar M_x^{(n)}(t=kT),\bar M_y^{(n)}(kT) \big)$ with $k=1,\cdots ,N_p$.


%
\begin{figure}[h]
	\includegraphics[width=\linewidth]{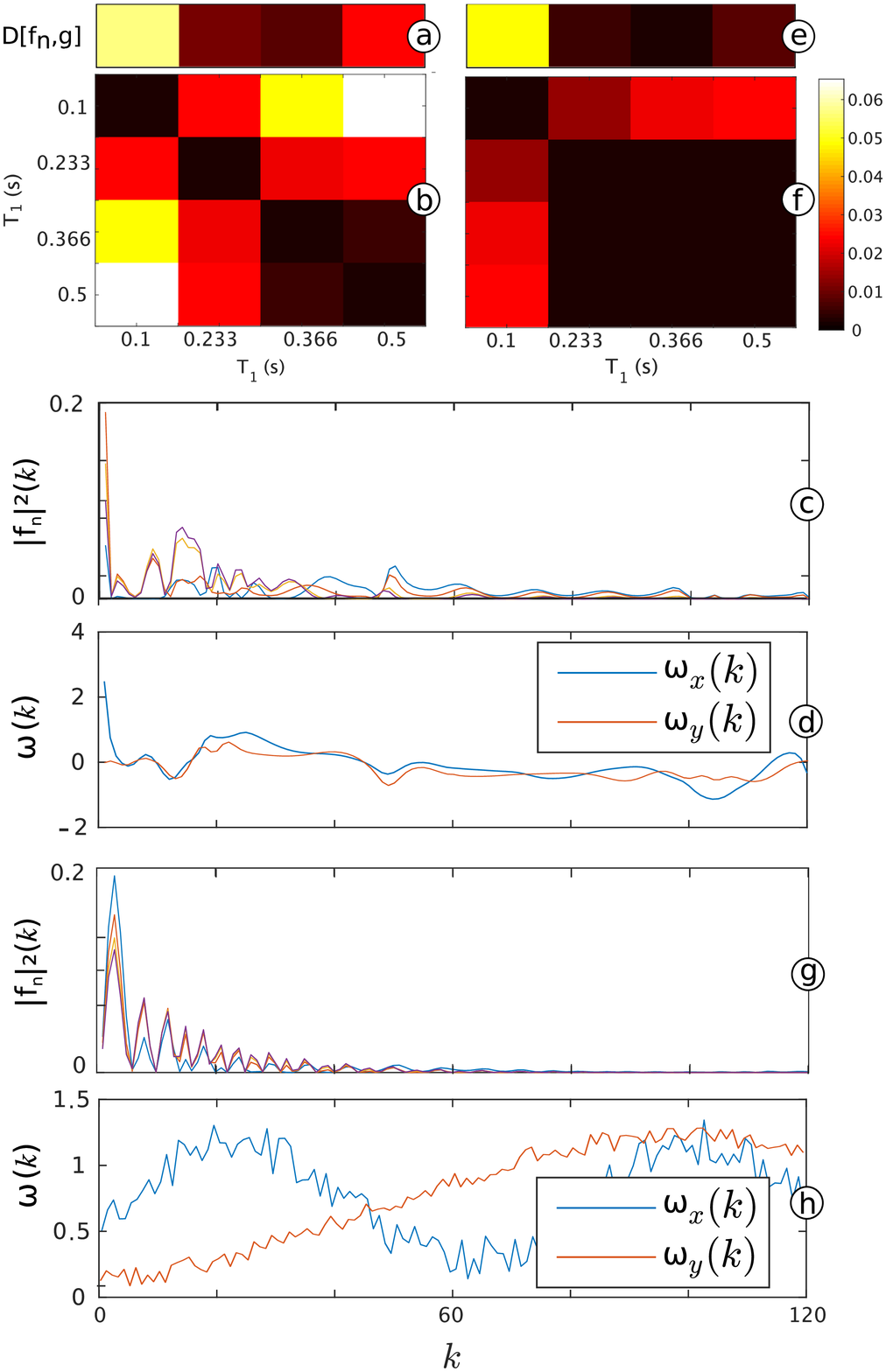}
	\caption{(Color online) The dictionary is composed of 4 elements regularly distributed in the interval $T_1\in[0.1,0.5]$~s. Giving an optimal (d) and a random (h) control fields (the black and the red (dark gray) lines represent respectively $\omega_x$ and $\omega_y$), we can compute the associate square modulus of dictionary functions (c) and (g). The dictionary functions are dimensionless. The efficiency can be checked with the  recognition maps $(T_1(m),T_1(n))\mapsto D[f_m, f_n]$ in (b) and (f). The panels (a) and (e) show the distance between the elements of the dictionary and the system to identify ($T_1 = 0.3~s$). The parameter $k$ refers to the number of pulses in the control process.}
	\label{fig:recognition_map}
\end{figure}
We first analyze the ideal situation of a homogeneous ensemble of spin- 1/2 particles irradiated on resonance, which is described by Eq.~\eqref{eq:bloch} with $\omega=0$ and $\alpha=1$. We assume that $T_2=0.2$~s is perfectly known, the goal being to estimate the value of $T_1=0.3$~s. To simplify the presentation of the results, we consider a simple database associated with four values of $T_1$: 0.1, 0.233, 0.366 and 0.5~s. Following the general procedure of OFP, we compute the optimal field for this dictionary in the case where all the coefficients $\mu_{mn}$ are set to 1. The time $T$ is set to 10~ms. The efficiency of the optimal solution is benchmarked against a time-dependant random field as shown in Fig.~\ref{fig:recognition_map}, which displays the recognition map $(T_1(m),T_1(n))\mapsto D[f_m,f_n]$ for the two databases and the time evolution of the different elements of the dictionary. The contrast of Fig.~\ref{fig:recognition_map} gives a first quantitative measure of the precision of the recognition process. In this example, $C_N$ is equal to 0.06 for the optimal fields and 0.03 under the random fields. The minimum values of the recognition maps are respectively 0.019 and 0.001. 

A first estimation of the value of $T_1$ can be made directly with the colorbars of Fig.~\ref{fig:recognition_map} and leads to $T_1\simeq 0.366$~s. Better accuracy of the fingerprinting process can be obtained by increasing the size of the dictionary. However, this procedure has a limit in terms of computational time, in particular to find the global optimum of the problem since the complexity of the control landscape increases rapidly with $N$. These numerical difficulties inherent to OFP can be avoided by using curve fitting in the post-measurement lookup stage. The fit is made with a minimization of $D$ based on a descent gradient algorithm with respect to the parameters $\vec{S}$ (here $T_1$ and $T_2$). In this case, the control field is fixed and a discrete derivative is used to compute the gradient $(\partial_{T_1}D[f(T_1,T_2),g],\partial_{T_2}D[f(T_1,T_2),g])$. Numerical simulations reveal that this approach converges after 50 or 100 iterations. Note that this concept is close to the Levenberg-Marquardt Method~\cite{Marquardt}, which is included in most of the curve fitting codes. For the ideal system, we obtain $T_1=0.3$~s both for the optimal and the random fields.

\begin{figure}[h]
\includegraphics[width=\linewidth]{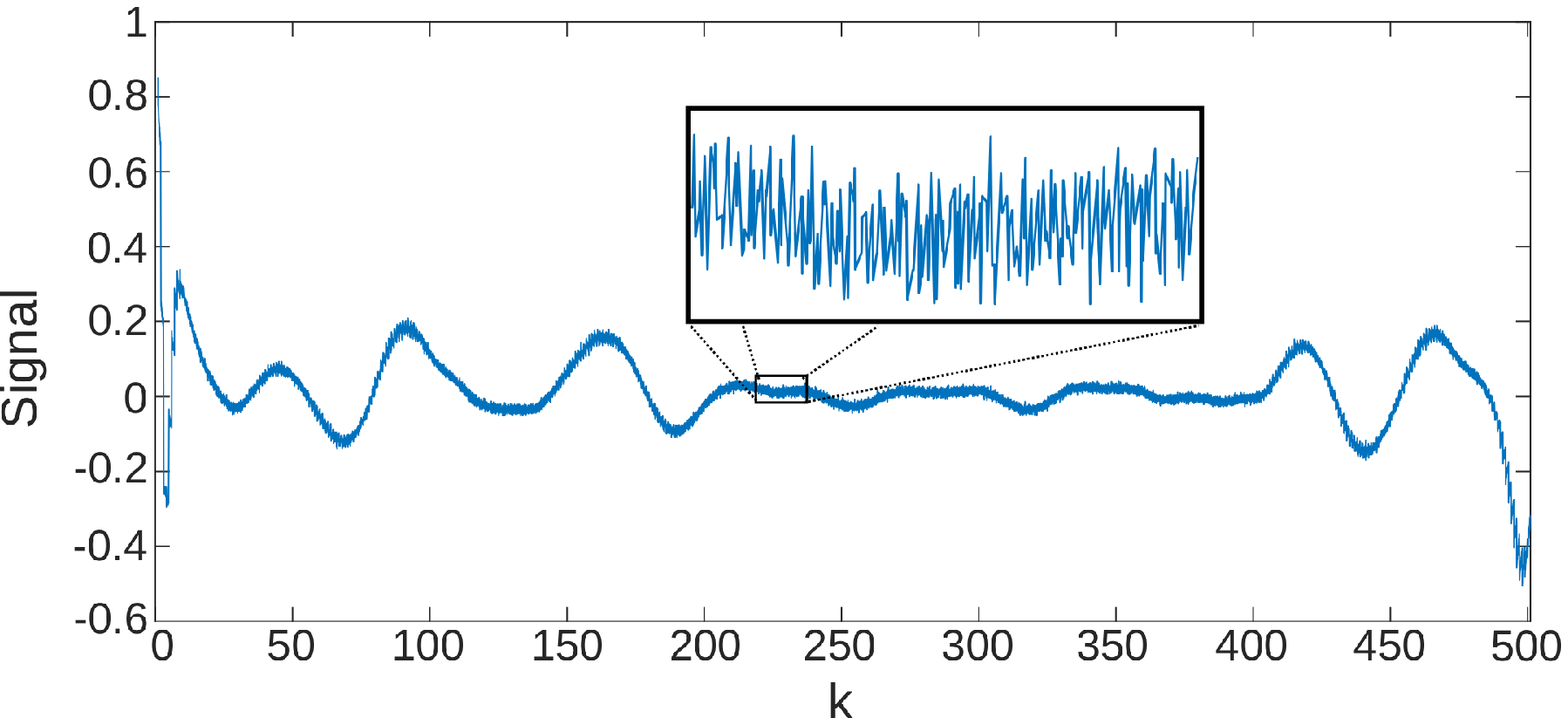}
\caption{(Color online) Time evolution of the experimental signal (in the $y$- direction) during a sequence of 500 $\delta$- pulses. The signal is expressed in arbitrary units. The insert is a zoom showing the fluctuations of the signal. The parameter $k$ refers to the number of the pulse in the control process.}
	\label{fig:noise}
\end{figure}
We investigate the stability in presence of noise of this approach. An experimental example is displayed in Fig.~\ref{fig:noise}, where we observe the fluctuations of the signal around a mean value. The experimental setup is modeled by considering a simulated noise added to the response of the system:
\begin{equation}
g(t)=\bar{g}+\epsilon\mathcal{N}(t),
\end{equation}
where $\bar{g}$ is the mean value of $g$ over many measurements, $\epsilon$ the standard deviation and $\mathcal{N}$, a gaussian noise centered in 0 with a variance of 1. Since the radius of the Bloch ball is normalized to 1, $\epsilon$ can be interpreted as a percent deviation. Using the optimal and random fields of Fig.~\ref{fig:recognition_map}, we optimize the parameter $T_1$ for different responses $g(t)$. The algorithm converges towards different values of $T_1$ for each response of the system. Figure~\ref{fig:uncertainty_fx_noise} displays the mean value and the width of the $T_1$- distribution (denoted $\Delta T_1^{\textrm{OPT}}$ and $\Delta T_1^{\textrm{RAND}}$ for the optimal and random fields respectively) as a function of $\epsilon$. For each value of $\epsilon$, we consider 30 measurements  $g(t)$ and the widths are determined by assuming a Gaussian distribution. This width can be interpreted as the accuracy of the corresponding estimation process. We observe in Fig.~\ref{fig:uncertainty_fx_noise} that the gain can be very large with the optimization procedure, a factor of the order of 100 for $\epsilon=0.001$ is obtained. The random field fails to predict $T_1$ accurately, even for low noise amplitude. In a standard experiment, the amplitude of the noise is generally of the order of 1\% of the maximum of the signal. This correspond here to $\varepsilon = 0.01$. Similar results have been obtained for the parameter $T_2$.

\begin{figure}[h]
\includegraphics[width=\linewidth]{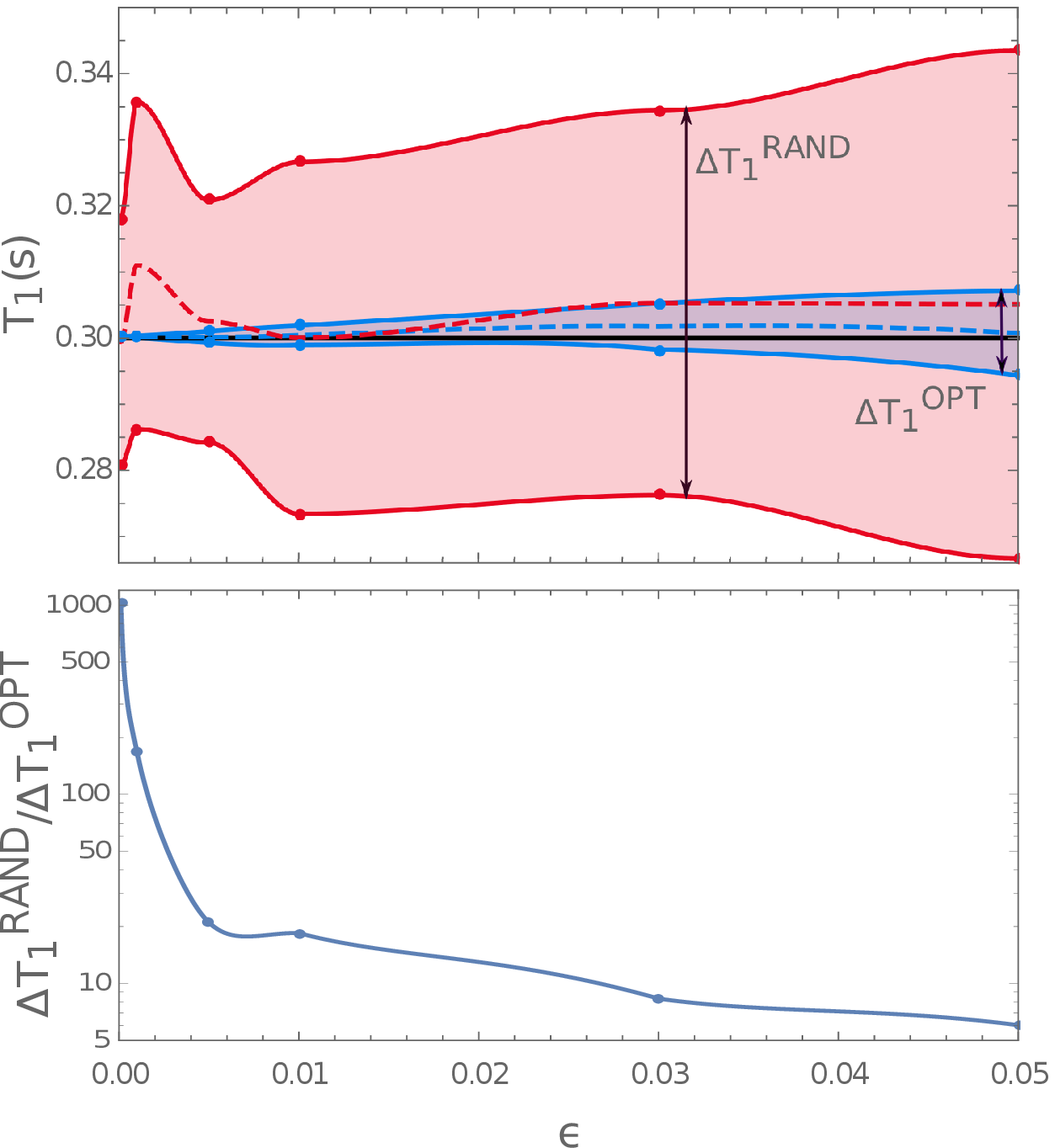}
\caption{(Color online) (top) Width of the distribution of the estimated $T_1$ values by the fingerprinting process (dark gray or blue - optimal, light gray or red - random) as a function of the noise amplitude $\epsilon$, which is dimensionless. The dashed lines depict the mean values of the two distributions. The horizontal solid line is the value of the $T_1$ parameter. (bottom) Plot of the ratio $\Delta T_1^{\textrm{RAND}}/\Delta T_1^{\textrm{OPT}}$ of the width of the two distributions as a function of $\epsilon$.}
	\label{fig:uncertainty_fx_noise}
\end{figure}

\section{Experimental results}\label{sec4}
We study experimentally the simultaneous estimation of the relaxation time $T_2$ and the distribution parameters of the offset inhomogeneities, while the $\alpha$- parameter can be set to 1 with good accuracy (see Eq.~\eqref{eq:bloch}). The offset distribution $\rho(\omega)$ is assumed to be Lorentzian:
\begin{equation}
\rho(\omega)\propto \left( 1+\frac{4(\omega - \bar \omega)^2}{\Delta \omega ^2}\right) ^{-1},
\end{equation}
where $\Delta \omega$ is the full width at half maximum (FWHM) and $\bar \omega$ the center of the distribution. The parameter $T_1$ was previously estimated to be 87~ms by inversion recovery \cite{levittbook}. The estimation of the parameter $T_2$ is a challenging issue because $T_2$ and $\Delta\omega$ are both responsible for the decay of the measured transverse magnetization. An effective transverse relaxation time $T_2^*$ defined by the relation
\begin{equation}\label{eqT2star}
\frac{1}{T_2^*}=\frac{1}{T_2} + \frac{\Delta \omega}{2}
\end{equation}
is usually introduced in magnetic resonance to account for the two physical effects \cite{levittbook}.
%
\begin{figure}
\includegraphics[width=\linewidth]{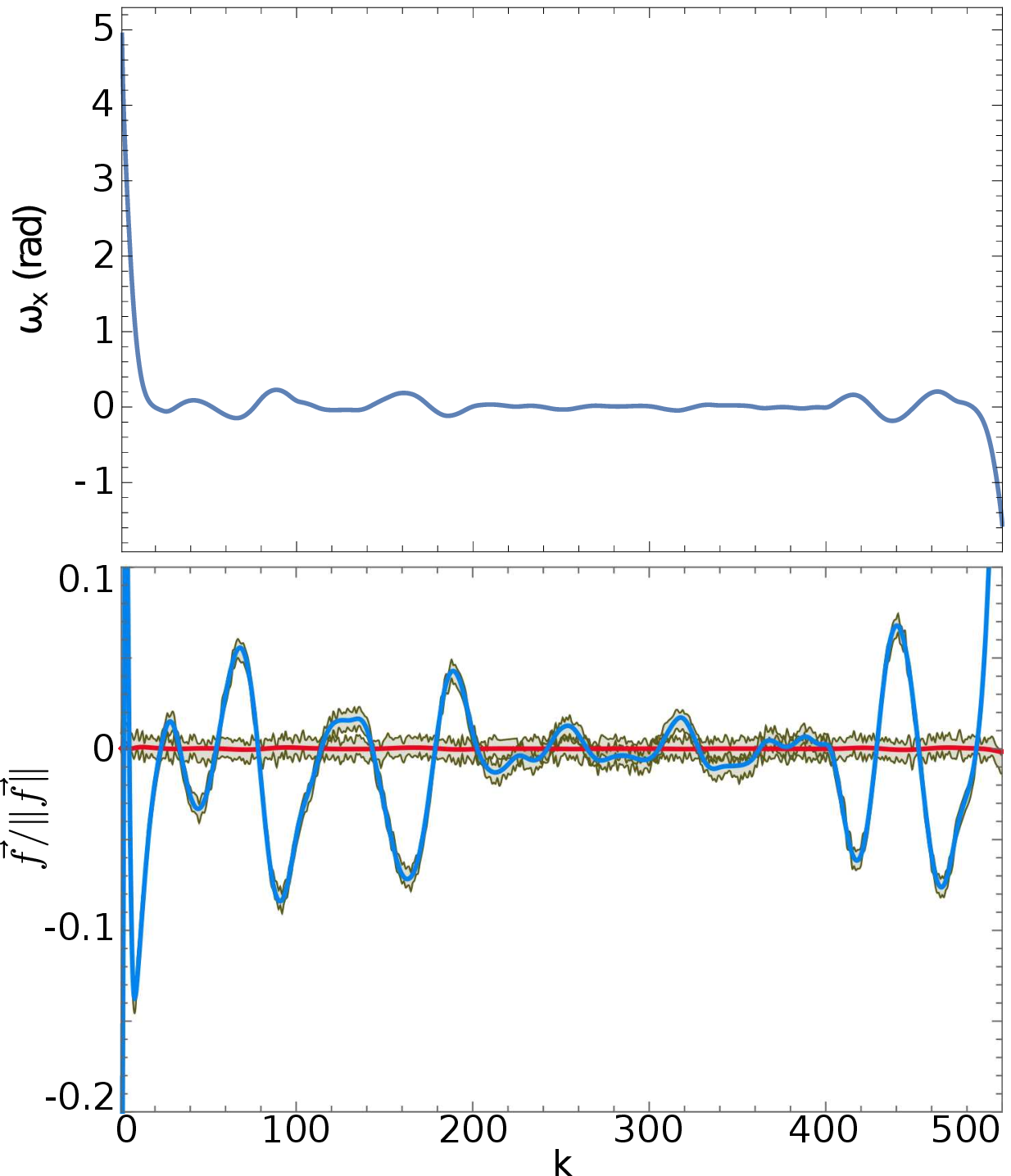}
\caption{(Color online) (top) Optimized control field along the $x$- direction ($\omega_y = 0$) to estimate the parameter $T_2$. (bottom) Time evolution of the simulated trajectories $M_x(t)$ (red or light gray) and $M_y(t)$ (blue or dark gray). The experimental data correspond to the gray areas around the numerical solutions, which give an estimation of the accuracy of the measurement (see the text for details).}
	\label{fig:experiment}
\end{figure}

A specific optimal pulse sequence sensitive to $T_2$ for an ensemble of spins with an average value of $\Delta\omega=20$~rad.s$^{-1}$ was designed. Note that only one control field along the $x$- direction was used to improve the convergence of the algorithm. Experiments were performed at room temperature on a Bruker Avance 600~MHz spectrometer. We used the $^1$H spins of H$_2$O with D$_2$O (99.9\%) as a solvent in a Shigemi tube. CuSO$_4$ was added as a $T_1$-shortening agent. The control field is a sequence of $N_p=500$ $\delta$-pulses separated by a time $T=10$~ms. The control field and the time evolution of the transverse magnetization are plotted in Fig.~\ref{fig:experiment}.
A reasonable match is found between the theoretical and the experimental results, which can be compared with the experimental error made in the measurement of the Bloch vector, as shown in Fig~\ref{fig:experiment}. Independent measurement based on a spin echo sequence leads to $T_2 = 60.5 \pm 0.5$~ms and $\Delta\omega=28.5 $~rad.s$^{-1}$. If we assume that the value of $\Delta \omega$ is known then OFP gives $T_2=60.4\pm 3.6$~ms and $\bar \omega = 0.1 \pm 0.6$~rad.s$^{-1}$. In the general case, due to the correlations between $\Delta \omega$ and $T_2$, it was not possible to estimate precisely the two parameters. As displayed in Fig.~\ref{fig:experiment2}, we observe that the figure of merit $D$ is almost the same for $\Delta\omega\in [20,38]$~rad.s$^{-1}$. On this interval, the value of $T_2^*$ is constant and in agreement with the experiment. Additional information would be required to estimate $T_2$ independently of $\Delta\omega$. From a computational point of view, it seems difficult to include different values of the bandwidth in the definition of the dictionary for improving the accuracy of the estimation. The same analysis was performed with several random sequences and we were not able to recover the right values of $T_2$ or $T_2^*$, showing thus the efficiency of OFP.

\begin{figure}
\includegraphics[width=\linewidth]{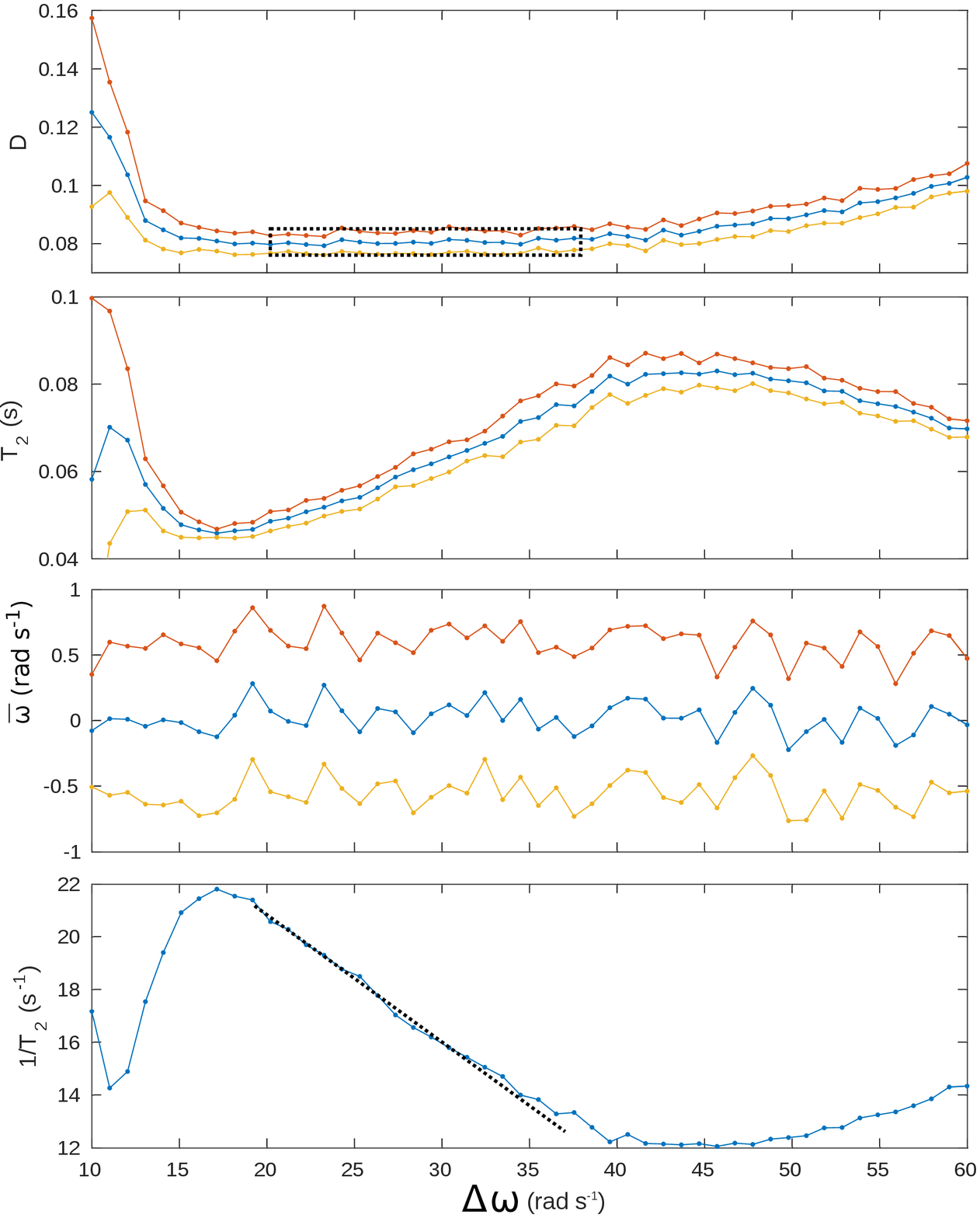}
\caption{(Color online) (Top) Minimum distance $D$ (defined in Eq.~\eqref{cost1}, $D$ is dimensionless) between the simulated and the experimental curves as a function of $\Delta \omega$. The dashed rectangle indicates the interval where $D$ is minimum. (Middle) Evolution of $T_2$ and $\bar{\omega}$ as a function of $\Delta\omega$. In the different panels, the blue (black) and the red/yellow (dark gray/ light gray) curves represent respectively the mean value of the signal and the upper/lower bounds of uncertainty. (Bottom) Plot of $1/T_2$ as a function of $\Delta \omega$. In the interval $[20,38]$~rad.s$^{-1}$, the parameter $T_2^*$ is constant as shown by the dotted line.} \label{fig:experiment2}
\end{figure}

\section{Conclusion}\label{conc}
We have introduced in this work the principles of OFP with an application to spin dynamics. The optimization procedure provides a method to approach the physical limits of the protocol in terms of sensitivity. FP has several advantages over the conventional methods. It allows a quantitative estimation of multiple parameters at the same time (e.g. the times $T_1$ and $T_2$), while only information about a single parameter is traditionally achieved. This advantage must be tempered if several parameters are correlated. This aspect has been illustrated in Sec.~\ref{sec4} with the offset terms and the $T_2$ relaxation time. The repeated acquisitions of data for the standard techniques are replaced by a single-shot measurement process in FP, which can drastically reduce the overall time of the experiment~\cite{naturemri}. Finally, FP is expected to be less sensitive to experimental imperfections and to the presence of noise. All these aspects are improved by the optimization procedure proposed in this paper. As shown in Sec.~\ref{sec3}, the better stability of OFP against noise perturbation is illustrated in a model example. A quantitative comparison with the standard Inversion Recovery approach~\cite{levittbook} is also made in App.~\ref{appc} for the estimation of the $T_1$ parameter. This theoretical comparison shows the better accuracy of OFP in this case.

This analysis paves the way for further investigations in MRI and realistic in vivo experiments~\cite{MRIbook,MRIconolly,lapertcontrast} in which the standard version of the fingerprinting process with random pulses has been applied with success~\cite{naturemri}. The concept of OFP could also be transferred to other domains such as quantum optics and atomic and molecular physics. An example is given by the control of molecular alignment and orientation in which pulse shaping techniques have been applied with success \cite{stapelfeldt:2003,seideman:2006}. The measure of the alignment could be used to estimate molecular parameters such as, e.g., the collisional relaxation rates \cite{ramakrishna:2005,vieillard:2013,vieillard:2008}. Another aspect could be to explore the applicability of this approach in a dynamical feedback framework where the control field would be adjusted in real time according to the results of the measurements. Finally, it seems promising to combine this technique with other methods of data analysis such as filter-diagonalization \cite{filter}, Bayesian estimation \cite{schirmer1} or Fisher  information~\cite{jacobpaper}.


\noindent\textbf{ACKNOWLEDGMENT.}\\
S.J. Glaser acknowledges support from the DFG (Gl 203/7-1). D. Sugny and S. J. Glaser acknowledge support from the ANR-DFG research program Explosys (ANR-14-CE35-0013-01; DFG-Gl 203/9-1). This work was done with the support of the Technische Universit\"at M\"unchen – Institute for Advanced Study, funded by the German Excellence Initiative and the European Union Seventh Framework Programme under grant agreement 291763. Experiments were performed at the
Bavarian NMR center at TU M\"unchen.

\appendix

\section{A mathematical description of the optimal fingerprinting process}\label{appa}
Optimal Fingerprinting Process (OFP) is not a standard problem in optimal control theory. The originality lies in the fact that the goal of the control procedure is not to bring a system from an initial state to a target state while minimizing a cost function. OFP is aimed at improving the characteristics of a dictionary in order to accurately estimate the physical parameters of a given system. The optimization in OFP is based on the maximization of a figure of merit, which is chosen in relation to the recognition process, i.e. the procedure used to find the best match between the measurements and the dictionary entries.

\subsection{The recognition process}\label{sec:recognition_process}
We consider three different sets:
\begin{itemize}
	\item a space of real square integrable functions $g$, with $g : [0,T] \mapsto \mathbb{R}^d$. This space is the set of all the possible measurements and $d$ the number of components of $g$.
	\item a set $\mathcal{S}$ of $N$ elements $\vec{S}_n $. Each $\vec{S}_n$ is a $p$- tuple of values of the $p$ physical parameters to estimate, $\vec{S}_n=(S_1(n),\cdots ,S_p(n))$.
	\item a set of $N$ time-dependent functions $f_n$ : $\{f_n \}_{n=1,\cdots , N}$, $f_n: [0,T] \mapsto \mathbb{R}^d$. This set is the dictionary used in a fingerprinting process.
\end{itemize}
The space of possible measurements is partitioned into $N$ different subsets, $\{\sigma_n\}_{n=1,\cdots, N}$. A function $f_n$ of the dictionary is associated with each element $\sigma_n$. Ideally, the partitioning satisfies the constraints:

 \textbf{1)} $f_n \in \sigma_n$ and $\forall k \neq n $, $f_n \not\in \sigma_k$,

 \textbf{2)} $\forall g \in \sigma_n$, $D[f_n,g]<D[f_k,g]$. If $D[f_n,g]=D[f_k,g]$, then $g$ belongs to the common boundary between $\sigma_n$ and $\sigma_k$.

%
%
The functional $D$ is defined by:
\begin{equation}
D[f_n,g] = \left \Vert \frac{f_n}{\Vert f_n\Vert} - \frac{g}{\Vert g\Vert} \right \Vert ^2.
\end{equation}
$D$ is the square of the distance between two normalized functions. The first function, $f_n$, belongs to the dictionary, and $g$, is the result of a measurement. $\Vert f\Vert$ refers to the norm of $f$, defined by $\Vert f\Vert =\sqrt{(f,f)}$, $(\cdot,\cdot)$ being the scalar product. In the continuous case, the scalar product of two functions $f(t)$ and $g(t)$ can be defined as $(f(t),g(t))=\int_0^Tf(t)g(t)dt$.
\begin{definition}
The recognition process consists in associating a function $g$ to an element $\sigma_n$ of the partition. To each $\sigma_n$ is attached a set of values of the physical parameters $\vec{S}_n$. A bijection can thus be defined between a partition and a specific physical system.
\end{definition}
The recognition process can be mathematically defined as follows
\begin{equation}\label{eqrp}
f_m =\arg \left[ \min_{n=1,\cdots, N} \left( D[f_n,g] \right)\right].
\end{equation}
Equation~\eqref{eqrp} means that the function $f_m$ associated with $g$ is the one minimizing $D[f_n,g]$ over all the possible functions $f_n$.
Note that the functional $D$ can also be written as:
\begin{equation}\label{eqdef2}
D[f_n,g] = 2\left( 1 - \frac{(f_n,g)}{\Vert f_n\Vert.\Vert g\Vert} \right).
\end{equation}
%
\subsection{Figure of merit and optimization}
\label{sec:fig_of_merit}
It is worth noting that normalized functions are used in the definition of the functional $D$. This point is due to the fact that the measurement process is defined up to a scaling factor. From a mathematical point of view, this means that we do not consider a function but a class of functions. This paragraph is aimed at giving a rigorous framework to this issue. The main result is the simplification of the figure of merit from a functional to a real function of one real variable. This geometric description gives  also an upper bound on the values of the figure of merit. We first define the equivalence classes of the functional $D$.
\begin{definition}\label{def2}
Two functions $g$ and $g'$ are said to be equivalent and denoted $g \sim_n g'$ if and only if  $D[f_n,g] = D[f_n,g']$.
The equivalence class is given by $\mathcal{C} _n ^\alpha = \{ g,~D[f_n,g] = 2\left( 1 - \cos\alpha_n \right)\}$, where $\alpha_n$ is the angle between $g$ and $f_n$, with $(g,f_n)=\Vert g\Vert \Vert f_n\Vert \cos\alpha_n $ (see Eq.~\eqref{eqdef2}).
\end{definition}

Note that the use of equivalence classes transforms the functional $D[f_n,g]$ defined over an infinite dimensional space into a simple function $D(\alpha_n)$ over $\mathbb{R}$ and the only relevant parameter is the angle $\alpha_n$.

The main difficulty from the optimization point of view is to define the concept of a \emph{good dictionary}. In particular, the size of the dictionary is arbitrary and depends on the discretization used for the physical parameters. Since the parameters take their values in a continuous set, it is possible to consider a dictionary of arbitrarily large size. Furthermore, the dictionary must be independent of experimental imperfections because it is computed before the measurement process. We solve this problem with the following proposition: The best dictionary is the one which optimizes the recognition process. We introduce the following figure of merit $C_N$ to measure the quality of the dictionary.
\begin{definition}
The figure of merit $C_N$ for a dictionary of $N$ elements is given by the mean value of all possible square distances between the functions $f_n$ and $f_k$:
\begin{equation}
C_N = \frac{1}{2 N^2}\sum_{n,k = 1}^{N} D[f_n,f_k].
\end{equation} \label{def:figure_of_merit}
\end{definition}
As shown below, the normalization factor $2N^2$ is chosen so that the upper bound of $C_N$ is 1. Some properties of $C_N$ can be established. Equivalent classes allow us to formulate the problem into a simple geometric picture. The normalized functions $\vec{\mathsf{f}}_n=f_n/\Vert f_n \Vert$ can be viewed as points belonging to a $(N-1)$- sphere $S^{N-1}$ of radius 1, and consequently the dictionary is a $(N-1)$- simplex. The distance between two vertices $\vec{\mathsf{f}}_n$ and $\vec{\mathsf{f}}_k$ is given by $\sqrt{D[\vec{\mathsf{f}}_n,\vec{\mathsf{f}}_k]}$. We are interested in the shape of the simplex which maximizes $C_N$ (i.e which maximizes the sum of the lengths of the edges). For $N=2$, it is obvious that the maximum is reached when $\vec{\mathsf{f}}_1 = -\vec{\mathsf{f}}_2$. For the case $N=3$, it can be shown that the highest value of $C_3$ is obtained for an equilateral triangle where each angle $\alpha_{nk}$, with $(\vec{\mathsf{f}}_n,\vec{\mathsf{f}}_k)=\cos (\alpha_{nk})$, is equal to $2\pi/3, ~\forall n,k$. For higher values of $N$, we get a regular simplex:
\begin{theorem}\label{thsimplex}
The optimal simplex is given by the set of functions $\{\vec{\mathsf{f}}_n\}_{n=1,\cdots,N}$ corresponding to a $(N-1)$- regular simplex of radius 1. The upper bound of $C_N$ is equal to 1.
\end{theorem}
\begin{proof}
We consider $\vec{\mathsf{f}}_n$ as a vector going from the center $O$ to a point of the hypersphere of radius 1. We have:
\[
\sum_{n,k=1}^N  \Vert \vec{\mathsf{f}}_n - \vec{\mathsf{f}}_k \Vert ^2  = 2N^2- 2\sum_{n,k=1}^N (\vec{\mathsf{f}}_n,\vec{\mathsf{f}}_k).
\]
Since
\[
\Vert \sum_{n=1}^N \vec{\mathsf{f}}_n \Vert ^2 = \sum_{n,k=1}^N(\vec{\mathsf{f}}_n,\vec{\mathsf{f}}_k),
\]
we deduce that
\[
\sum_{n,k=1}^N  \Vert \vec{\mathsf{f}}_n - \vec{\mathsf{f}}_k \Vert ^2 = 2N^2 - 2\Vert \sum_{n=1}^N \vec{\mathsf{f}}_n \Vert ^2.
\]
This expression shows that the maximum value of $\sum_{n,k}  \Vert \vec{\mathsf{f}}_n - \vec{\mathsf{f}}_k \Vert ^2$ is $2N^2$, i.e. the inverse of the normalization factor of $C_N$. We obtain that $C_N\leq 1$. The maximum value is reached for $\sum_{n=1}^N \vec{\mathsf{f}}_n =0$. This is the equation of the simplex barycenter, which is equal to zero for a regular simplex \cite{simplex}.
\begin{flushright}
$\square$
\end{flushright}
\end{proof}

\section{Numerical optimal control algorithm}\label{appb}

This paragraph is aimed at briefly describing the extended version of GRAPE used in the numerical simulations \cite{grape}. This extension is closely related to the concept of optimal tracking introduced in Ref.~\cite{trackinggrape}, where the goal of the control is to steer the evolution of the system along a specified trajectory. To simplify the description of the algorithm, we consider here the dynamics of the spins in the $(y,z)$- plane. We recall that the dynamics can be integrated step by step as follows. We denote by $\vec{M}=(y,z)$ the state of the system.
After a $\delta$- pulse and a free relaxation, we have:
\begin{equation}
\vec{M}_1=\mathcal{L}_T\mathcal{R}_\theta\vec{M}_0
\end{equation}
where
\begin{eqnarray*}
\mathcal{R}_\theta=\left(\begin{array}{cc} \cos\theta & -\sin\theta \\
\sin\theta & \cos\theta\end{array}\right),
\end{eqnarray*}
and
\begin{eqnarray*}
\mathcal{L}_T\vec{M}=\left(\begin{array}{c} 0 \\
1-e^{-T/T_1} \end{array}\right)+
\left(\begin{array}{cc} \exp[-T/T_2] & 0 \\
0 & \exp[-T/T_1] \end{array}\right)\vec{M},
\end{eqnarray*}
with $\theta$ the angle of the $\delta$- pulse and $T$ the time between two $\delta$- pulses.

After $N$ processes, we get:
\begin{equation}
\vec{M}_N=\mathcal{L}_{T_N}\mathcal{R}_{\theta_N}\mathcal{L}_{T_{N-1}}\mathcal{R}_{\theta_{N-1}}\cdots
\mathcal{L}_{T_{1}}\mathcal{R}_{\theta_{1}}\vec{M}_0
\end{equation}
where the parameters $\theta_k$ have to be adjusted to maximize a given figure of merit $\Phi$ and the times $T_N=T_{N-1}=\cdots = T_1$ are fixed. We define the control field as $\vec{\theta}=(\theta_1,\theta_2,\cdots,\theta_N)$. A standard GRAPE algorithm can then be used to maximize $\Phi$.

Let us assume for instance that $\Phi=\Phi(\vec{M}_N)$. The gradient can be written as follows:
\begin{equation}
\frac{\partial \vec{M}_N}{\partial \theta_k}=\mathcal{L}_{T_N}\mathcal{R}_{\theta_N}\mathcal{L}_{T_{N-1}}\mathcal{R}_{\theta_{N-1}}\cdots d\mathcal{R}_{\theta_k} \cdots
\mathcal{L}_{T_{1}}\mathcal{R}_{\theta_{1}}\vec{M}_0
\end{equation}
where
\begin{eqnarray*}
d\mathcal{R}_\theta=\left(\begin{array}{cc} -\sin\theta & -\cos\theta \\
\cos\theta & -\sin\theta\end{array}\right).
\end{eqnarray*}
Introducing the adjoint state $N_k$ such that:
\begin{equation}
\vec{N}_k=\mathcal{L}^{-1}_{T_{k}} \cdots \mathcal{R}^{-1}_{\theta_{N-1}}\mathcal{L}^{-1}_{T_{N-1}}\mathcal{R}^{-1}_{\theta_N}\mathcal{L}^{-1}_{T_N}\vec{M}_N,
\end{equation}
we get:
\begin{equation}
\frac{\partial \vec{M}_N}{\partial \theta_k}=\vec{N}_k d\mathcal{R}_{\theta_k}\vec{M}_{k-1}.
\end{equation}
At each step of the algorithm, the field $\vec{\theta}$ is corrected as follows:
\begin{equation}
\vec{\theta}\to \vec{\theta}+\varepsilon \frac{\partial \Phi}{\partial \vec{\theta}},
\end{equation}
where $\varepsilon$ is a small parameter chosen to ensure the increase of the figure of merit and
\begin{equation}
\frac{\partial \Phi}{\partial \vec{\theta}}=\frac{\partial \Phi}{\partial \vec{M}_N}\frac{\partial \vec{M}_N}{\partial \vec{\theta}}.
\end{equation}

In the fingerprinting procedure, we consider the case of two spins to simplify the notations. The two spins are respectively described by the coordinates $y$ and $\tilde{y}$, whose dynamics are given by:
\begin{equation}
\vec{M}_N=\begin{pmatrix} y_N \\ z_N \end{pmatrix}=U_N\cdots U_1 \begin{pmatrix} y_0 \\ z_0 \end{pmatrix}
\end{equation}
and
\begin{equation}
\vec{\tilde{M}}_N=\begin{pmatrix} \tilde{y}_N \\ \tilde{z}_N \end{pmatrix}=U_N\cdots U_1 \begin{pmatrix} \tilde{y}_0 \\ \tilde{z}_0 \end{pmatrix},
\end{equation}
where $U_k(\theta_k)$ depends only on $\theta_k$, the $k$~th control parameter. The figure of merit $C$ to maximize is given by:
\begin{equation}
C=\frac{1}{2}\sum_{i=1}^N[y_i-\tilde{y}_i]^2,
\end{equation}
where the vectors are not divided by their norms for clarity purpose.
The gradient of $C$ with respect to $\theta_k$ can be written as:
\begin{equation}
\frac{\partial C}{\partial \theta_k}=\sum_{i\geq k}^N (y_i-\tilde{y}_i)(\frac{\partial y_i}{\partial \theta_k}-\frac{\partial \tilde{y}_i}{\partial \theta_k}).
\end{equation}
Since:
\begin{equation}
\frac{\partial y_i}{\partial \theta_k}=\partial_k y_i=U_i\cdots U_{k+1}\partial_k U_k U_{k-1}\cdots U_1 \vec{M}_0|_y,
\end{equation}
we have:
\begin{equation}
\partial_k C=\sum_{i\geq k}^N (y_i-\tilde{y}_i)U_i\cdots U_{k+1}\partial_k U_k U_{k-1}\cdots U_1 (\vec{M}_0-\vec{\tilde{M}}_0)|_y,
\end{equation}
which can also be written as:
\begin{equation}
\partial_k C=\sum_{i\geq k}^N [(y_i-\tilde{y}_i)U_i\cdots U_{k+1}]\partial_k U_k U_{k-1}\cdots U_1 (\vec{M}_0-\vec{\tilde{M}}_0)|_y.
\end{equation}
Introducing a generalized adjoint state $\mathcal{P}$ such that:
\begin{equation}
\mathcal{P}_k=\sum_{i\geq k}^N [{^t\begin{pmatrix}y_i-\tilde{y}_i \\ 0\end{pmatrix}}U_i\cdots U_{k+1}],
\end{equation}
we obtain:
\begin{equation}
\partial_k C=\mathcal{P}_k \partial_k U_k (\vec{M}_{k-1}-\vec{\tilde{M}}_{k-1})|_y.
\end{equation}
As in a standard GRAPE algorithm, the field $\vec{\theta}$ is corrected at each step of the algorithm as follows:
\begin{equation}
\vec{\theta}\to \vec{\theta}+\varepsilon \frac{\partial C}{\partial \theta_k}.
\end{equation}

\section{Comparison with the Inversion Recovery method}\label{appc}

We study in this paragraph the efficiency of OFP with respect to the Inversion Recovery approach (IR), which is a standard way to estimate the relaxation time $T_1$~\cite{levittbook}. IR is based on the successive application of a $\pi$- pulse followed by a series of $\pi/2$- pulses at different times to measure the transverse magnetization. A fair comparison between the two estimation techniques is difficult and heavily depends on the features of the experimental set-up. Here, we investigate the example analyzed in Sec.~\ref{sec3} to avoid such a discussion. For the IR, we consider a single-shot measurement process during a relaxation towards the thermal equilibrium state in which the longitudinal relaxation can be measured in an arbitrary short time with a noise added to the response of the system. Note that the waiting time between each acquisition is not included in this ideal approach, which overestimates the efficiency of a realistic IR. The response of the system is described as follows:
$$
g(t_m)=M_z(t_m)+\varepsilon \mathcal{N}(t_m),
$$
where the parameter $\varepsilon$ and the Gaussian noise $\mathcal{N}(t)$ are defined as in Sec.~\ref{sec3}. The time evolution of the longitudinal magnetization is given by a perfect inversion dynamics $M_z(t)=1-2\exp[-t/T_1]$ with 120  times $t_m$ separated by 10~ms. The same noise and the same number of measurement points are therefore used for OFP and IR, which ensures a fair comparison. The results are displayed in Fig.~\ref{figIR} and show that OFP has a better accuracy than IR. For a $T_1$ value of 300~ms and a noise amplitude $\varepsilon=0.05$, OFP achieves a precision of the order of $\pm 0.05$~ms, whereas the precision of IR is larger than 2~ms. A gain of a factor of 4 in estimating $T_1$ is obtained. Since there is no steady state in OFP, this factor is expected to increase for longer pulse sequences.
\begin{figure}
\includegraphics[width=\linewidth]{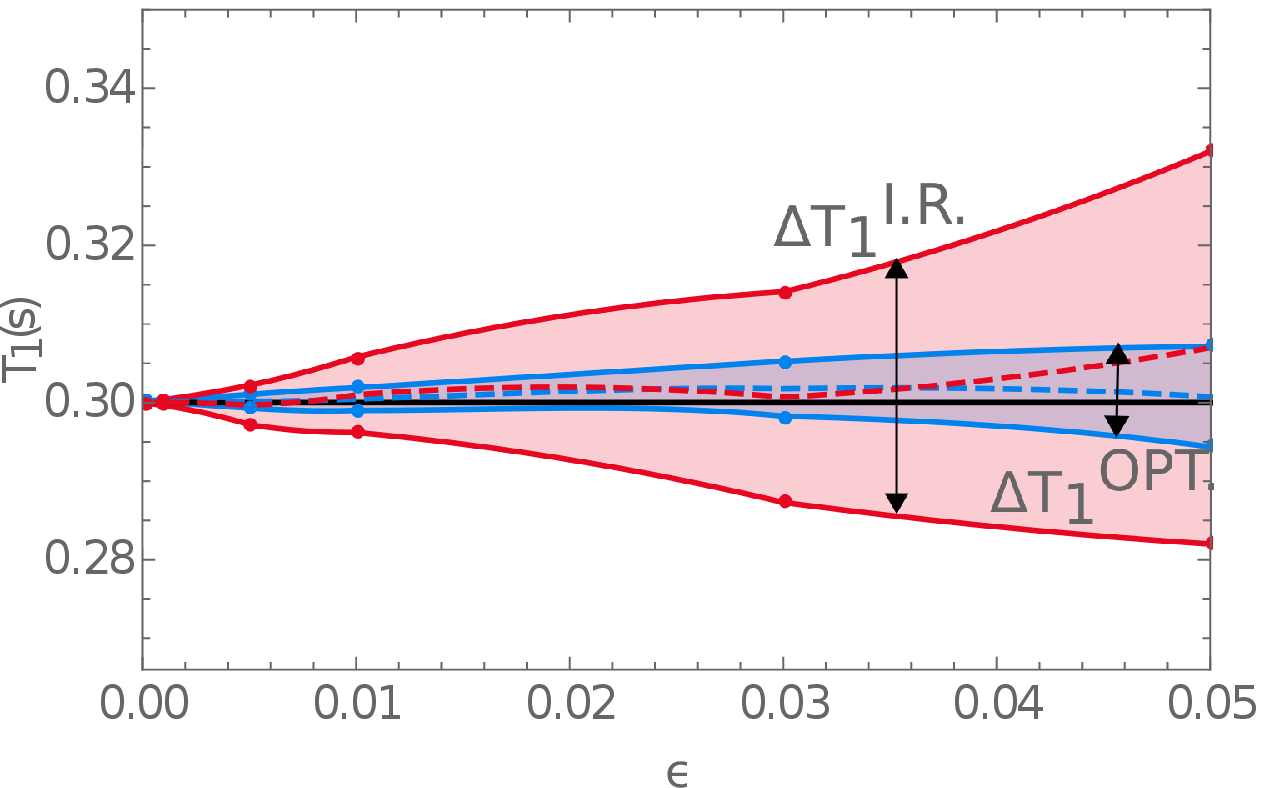}
\caption{(Color online) Width of the distribution of the estimated $T_1$ values by OFP (dark gray or blue) and IR (light gray or red) as a function of the noise amplitude $\epsilon$, which is dimensionless. The dashed lines depict the mean values of the two distributions. The horizontal solid line is the value of the $T_1$ parameter.} \label{figIR}
\end{figure}

\bibliographystyle{apsrev}%


\end{document}